\documentclass{sig-alternate}

\usepackage{amsthm,amssymb,amsmath}
\usepackage{graphics}
\usepackage{bbm}
\usepackage{tikz}
\usepackage[plainpages=false,pdfpagelabels,colorlinks=true,citecolor=blue,hypertexnames=false]{hyperref}
\usepackage{color}
\newtheorem{theorem}{Theorem}
\newtheorem{notation}[theorem]{Convention}

\newtheorem{corollary}[theorem]{Corollary}
\newtheorem{lemma}[theorem]{Lemma}
\newtheorem{remark}[theorem]{Remark}

\newtheorem{defi}[theorem]{Definition}
\newtheorem{example}[theorem]{Example}
\newtheorem{fact}[theorem]{Fact}
\def\qed{\quad\rule{1ex}{1ex}}

\newcommand{\bE}{ {\mathbb E}}

\newcommand{\cB}{ {\mathcal B}}

\newcommand{\bF}{ {\mathbb F}}
\newcommand{\bN}{ {\mathbb N}}
\newcommand{\bZ}{ {\mathbb Z}}
\newcommand{\cP}{ {\mathcal P}}
\newcommand{\cM}{ {\mathcal M}}
\newcommand{\cN}{ {\mathcal N}}

\newcommand{\pa}{\partial}
\newcommand{\spanning}{\text{span}}

\def\lc{\operatorname{lc}}

\begin{document}

\title{Hermite Reduction and Creative Telescoping \\ for Hyperexponential Functions \titlenote{S.C.\ was supported
by the National Science Foundation (NSF) grant CCF-1017217, A.B.\ and F.C.\ were supported in part by the MSR-INRIA Joint
Centre, Z.L.\ by two NSFC grants (91118001 and 60821002/F02), and G.X.\ by NSFC grant (11171231).}
}

\numberofauthors{1}

\author{\medskip
Alin Bostan$^1$, \, Shaoshi Chen$^{2}$, \, Fr\'ed\'eric Chyzak$^1$,  \, Ziming Li$^3$, \, Guoce Xin$^4$ \\
\smallskip
       \affaddr{$^1$INRIA, Palaiseau, 91120, (France)}\\
       \smallskip
       \affaddr{$^2$Department of Mathematics, NCSU, Raleigh, 27695-8025, (USA)}\\
       \smallskip
       \affaddr{$^3$KLMM,\, AMSS, \,Chinese Academy of Sciences, Beijing 100190, (China)}\\
       \smallskip
       \affaddr{$^4$Department of Mathematics, Capital Normal University, Beijing 100048, (China)}\\
       \smallskip
      \email{schen21@ncsu.edu, \, $\{$alin.bostan, frederic.chyzak$\}$@inria.fr}\\
      \email{zmli@mmrc.iss.ac.cn, \,  guoce.xin@gmail.com}
%Preliminary notes
}

\maketitle
\begin{abstract}
We present a reduction algorithm that simultaneously extends
Hermite's reduction for rational functions and the
Hermite-like reduction for hyperexponential functions. It yields a
unique additive decomposition and allows to decide
hyperexponential integrability. Based on this reduction algorithm,
we design a new method to compute minimal telescopers for bivariate
hyperexponential functions. One of its main features is that it can
avoid the costly computation of certificates. Its implementation
outperforms Maple's function {\sf{DEtools[Zeilberger]}}. Moreover,
we derive an order bound on minimal telescopers, which is more
general and tighter than the known one.
\end{abstract}



\keywords{{Hermite Reduction, Hyperexponential function, Telescoper}

%\bigskip

\section{Introduction}\label{SECT:intro}

Given a univariate rational function~$r$, Hermite reduction
in~\cite{Hermite1872, Horowitz1971, BronsteinBook} finds rational
functions~$r_1$ and~$r_2$ s.t.\ (i)~$r=r_1{+}r_2$, (ii)~$r_1$ is
rational integrable, (iii)~$r_2$ is a proper fraction with a
squarefree denominator. The additive decomposition is unique, and
$r$ is rational integrable if and only if~$r_2=0$.

A univariate function is hyperexponential if its logarithmic
derivative is rational. Exponential, radical and rational functions
are hyperexponential. Rational Hermite reduction has been extended
to hyperexponential functions by Davenport in~\cite{Davenport1986}
and by~Geddes, Le and Li in~\cite{GeddesLeLi2004}. The former aims
at solving Risch's equation; the latter is a differential analogue
of the reduction algorithm for hypergeometric terms
in~\cite{Abramov2001}. For a given hyperexponential function~$H$,
the reduction algorithms in~\cite{Davenport1986,GeddesLeLi2004} compute two hyperexponential
functions~$H_1$ and~$H_2$ s.t.~(i)~$H=H_1+H_2$, (ii)~$H_1$ is
hyperexponential integrable, (iii)~$H_2$ is minimal in some sense.
However,~$H_2$ is not unique in general and it may be nonzero even
when~$H$ is hyperexponential integrable. In order to decide the
integrability of~$H$, one additionally needs to compute polynomial
solutions of a first-order linear differential equation.

The method of creative telescoping for hyperexponential functions is
developed by Almkvist and Zeilberger in~\cite{Almkvist1990}.
It is nowadays an important automatic tool for computing definite integrals. Recently, it has
also played an important role in the resolution of intriguing
problems in enumerative combinatorics~\cite{KKZ2009, KKZ2011}. For a
bivariate hyperexponential function~$H(x, y)$, the problem of
creative telescoping is to find a nonzero operator~$L(x, D_x) {\in}
\bF(x)\langle D_x \rangle$, the ring of linear differential
operators over the rational-function field~$\bF(x)$,~s.t.
\begin{equation}\label{EQ:generalct}
L(x, D_x)(H) = D_y(G)
\end{equation}
for some hyperexponential function~$G$, where~$D_x = \pa/\pa x$ and
$D_y = \pa/\pa y$. The operator~$L$ above is called a
\emph{telescoper} for~$H$, and~$G$ is the corresponding {\em
certificate}. An algorithm for solving~\eqref{EQ:generalct} is given
in~\cite{Almkvist1990}, and is based on differential Gosper's algorithm. An algorithm for rational-function
telescoping is given in~\cite{BCCL2010}, and is based on
Hermite reduction. The latter separates the computation for telescopers from
that for certificates, and has a lower complexity than the former for rational functions.

In the present paper, we develop a reduction algorithm which, given
a univariate hyperexponential function~$H$, constructs
two hyperexponential functions~$H_1$ and~$H_2$ s.t.~(i) $H=H_1 + H_2$,
(ii)~$H_1$ is hyperexponential integrable, and (iii)~$H_2$ is either
zero or not hyperexponential integrable. We show that~$H_2$ in the
above additive decomposition  is unique and can be obtained without
computing polynomial solutions of any differential equation. Our
algorithm is based on the Hermite-like reduction
in~\cite{GeddesLeLi2004}, a differential variant of the
polynomial reduction in~\cite{Abramov2001}
and on the idea for reducing simple radicals in~\cite[Proposition~7]{XinZhang2009}.
The main new ingredients are the uniqueness of~$H_2$ and an easy way to compute~$H_2$,
which are crucial for many applications.
These enable us to extend the
reduction-based rational telescoping algorithm in~\cite{BCCL2010} to
the hyperexponential case, and derive an order bound on the
telescopers. The bound is more general and tighter than that given
in~\cite{Apagodu2006}.

The rest of the paper is organized as follows.
We review the notion of hyperexponential functions
and Hermite-like reduction in Sections~\ref{SECT:hfs} and~\ref{SECT:hermite-like}, respectively.
A new reduction algorithm is developed for hyperexponential functions in Section~\ref{SECT:hermiteH}.
After introducing kernel reduction in Section~\ref{SECT:kernel},
we present a reduction-based telescoping algorithm for bivariate hyperexponential functions, and derive an upper
bound on the order of minimal telescopers in Section~\ref{SECT:ct}.
We briefly describe an implementation of the new telescoping algorithm,
and present some experimental results in Section~\ref{SECT:timings}, which validate its practical relevance.

As a matter of notation, we let~$\bE$ be a field of characteristic zero and~$\bE(y)$ be the
field of rational functions in~$y$ over~$\bE$. For a polynomial~$p \in \bE[y]$, we denote by~$\deg(p)$ and~$\lc(p)$
the degree and leading coefficient of~$p$, respectively.
Let~$D_y$ denote the usual derivation~$d/dy$ on~$\bE(y)$. Then~$(\bE(y), D_y)$ is a differential field.

\section{Hyperexponential functions}\label{SECT:hfs}
Hyperexponential functions share the common properties of
rational functions, simple radicals,
and exponential functions. Together with hypergeometric terms,
they are frequently viewed as a special and important class of \lq\lq closed-form\rq\rq~solutions of linear
differential and difference equations with polynomial coefficients.

\begin{defi}\label{DEF:hf}
Let~$\Phi$ be a differential field extension of~$\bE(y)$. A nonzero element~$H\in \Phi$
is said to be \emph{hyperexponential} over~$\bE(y)$ if its logarithmic derivative~${D_y(H)}/{H}  \in \bE(y)$.
\end{defi}
The product of hyperexponential functions
is also hyperexponential. Two hyperexponential functions~$H_1, H_2$
are said to be \emph{similar} if there exists~$r\in \bE(y)$ s.t.~$H_1 = rH_2$.
The sum of similar hyperexponential functions is still hyperexponential, provided
that it is nonzero.

For brevity, we use the notation~$\exp(\int  f dy)$ to indicate a
hyperexponential function whose logarithmic derivative is~$f$. For a rational function~$r\in \bE$, we have
\[r \exp\left(\int  f\, dy\right) = \exp\left(\int  \left(f + D_y(r)/r\right)\, dy\right).\]
A univariate hyperexponential function~$H$ is
said to be \emph{hyperexponential integrable} if it is the derivative of another hyperexponential
function. For brevity, we say \lq\lq integrable\rq\rq~instead of \lq\lq hyperexponential integrable\rq\rq~in the sequel.

Assume that~$H = r\exp\left( \int f dy \right)$ is integrable. Then~$H$ is equal to~$D_y(G)$ for some hyperexponential function~$G$. A
straightforward calculation shows that~$G$ is similar to~$D_y(G)$, and so is~$H$. Set~$G= s \exp\left( \int f \right)$
for some~$s \in \bE(y)$. Then~$H=D_y(G)$  if and only if
\begin{equation} \label{EQ:integrable}
r = D_y(s) + f \, s.
\end{equation}
Deciding the integrability of~$H$ amounts to finding a rational solution~$s$ s.t.~the above equation holds.

\section{Hermite-like reduction}\label{SECT:hermite-like}
Reduction algorithms have been developed
for computing additive decompositions of rational functions~\cite{Ostrogradsky1845, Hermite1872, Horowitz1971},
hypergeometric terms~\cite{Abramov1975, Abramov2001}, and hyperexponential
functions~\cite{Davenport1986, GeddesLeLi2004}.
Those algorithms can be viewed as generalizations
of Gosper's algorithm~\cite{Gosper1978} and its
differential analogue~\cite[\S 5]{Almkvist1990}.

For a hyperexponential function~$H$, a reduction algorithm
computes two hyperexponential functions~$H_1, H_2$ s.t.
\begin{equation} \label{EQ:hrl}
H = D_y(H_1) + H_2.
\end{equation}
It turns out that~$H, H_1$ and~$H_2$ are similar. So we may write~$H {=} r \exp\left(\int f  dy\right)$
and~$H_i{=}r_i \exp\left(\int f  dy\right)$, where~$r, r_i, f$ belong to~$\bE(y)$ and~$i=1,2$.
Then~\eqref{EQ:hrl} translates into
\[ r = D_y(r_1) + f \,r_1  + r_2. \]
A reduction algorithm for computing~\eqref{EQ:hrl} amounts to choosing rational functions~$r, f$ and~$r_1$ so that~$r_2$ satisfies properties
similar to those obtained in Hermite reduction for rational functions.
There are at least two approaches to this end. One is given in~\cite{Davenport1986}, and the other in~\cite{GeddesLeLi2004}.
We review the latter, because  the notion of differential-reduced
rational functions plays a key role in Lemma~\ref{LM:poly} in Section~\ref{SECT:hermiteH}.

%\subsection{Differential rational canonical forms}
Recall~\cite[\S 2]{GeddesLeLi2004} that a rational function~$r=a/b\in \bE(y)$ is said
to be \emph{differential-reduced} w.r.t.~$y$ if
\[\gcd\left(b, a-i\, D_y (b) \right)= 1 \quad \text{for all $i\in \bZ$.}\]
By Lemma~2 in~\cite{GeddesLeLi2004},~$r$ is differential-reduced if and only if none of its residues is an integer.
The {\em differential rational canonical  form} of a rational function~$f$ in~$\bE(y)$ is
a pair~$(K, S)$ in~$\bE(y) \times \bE(y)$ s.t. (i)~$K$ is differential-reduced;
(ii)~the denominator of~$S$ is coprime with that of~$K$; and (iii)~$f$ is equal to~$K + D_y(S)/S$.
Every rational function has a unique canonical form in the sense that~$K$ is unique and~$S$ is unique up to
a multiplicative constant in~$\bE$~\cite[\S 3]{GeddesLeLi2004}.
We call~$K$ and~$S$ the \emph{kernel} and~\emph{shell} of~$f$, respectively.
They can be constructed by the method described in~\cite[\S 3]{GeddesLeLi2004}.

Let~$H$ be a univariate hyperexponential function in the form~$\exp(\int  f dy)$ over~$\bE$.
Assume that~$K$ and~$S$ are the kernel and shell of~$f$, respectively.
Then~$H = S \exp\left(\int K \, dy \right).$
Note that~$K = 0$ if and only if~$H$ is a rational function, which is equal to~$cS$ for some~$c \in \bE$.
\begin{example}\label{Exam:example1}
Let~$H = \sqrt{y^2+1}/(y-1)^2.$
The logarithmic derivative of~$H$  is
\[\frac{D_y H}{H} = \frac{D_y (1/(y-1)^2)}{1/(y-1)^2} + \frac{y}{y^2+1},\]
where~$y/(y^2+1)$ is differential-reduced.  The kernel and shell of~$D_y(H)/H$ are $y/(y^2+1)$ and~$1/(y-1)^2$,
respectively. So~$H=\exp\left(\int y/(y^2+1)\, dy \right)/(y-1)^2.$
\end{example}

For brevity, we make a notational convention.
\begin{notation} \label{CON:kernel}
Let~$H$ denote a hyperexponential function whose logarithmic derivative has kernel~$K$ and shell~$S$.
Assume that~$K$ is nonzero, that is,~$H$ is not a rational function. Set~$T=\exp\left( \int K \, dy \right)$.
Moreover, write~$K=k_1/k_2$,
where~$k_1, k_2$ are polynomials in~$\bE[y]$ with~$\gcd(k_1, k_2)=1$.
\end{notation}

The algorithm~{\bf ReduceCert} in~\cite{GeddesLeLi2004} computes a rational function~$S_1$ s.t.
\begin{equation}\label{EQ:certred}
S = D_y(S_1) + S_1 K + \frac{a}{b k_2},
\end{equation}
where~$a, b\in \bE[y]$ satisfy the following conditions:~$b$ is the squarefree part
of the denominator of~$S$, and $\gcd(b, k_2){=}1$.
Note that~$a$ is not necessarily coprime with~$bk_2$.
As the algorithm~{\bf ReduceCert} only reduces the shell~$S$,
it is referred to as the {\em  shell reduction}.
It follows from~\eqref{EQ:certred} that
\begin{equation} \label{EQ:shell}
H = D_y\left( S_1 T \right) + \frac{a}{b k_2} T.
\end{equation}
By Theorem~4 in~\cite{GeddesLeLi2004},~$a/b$ belongs to~$\bE[y]$ if~$H$ is integrable.
\begin{example}\label{EXAM:shell}
Let~$H$ be the same hyperexponential function as in Example~\ref{Exam:example1}.
Then~$D_y(H)/ H$ has kernel~$K = y/(y^2 + 1)$ and shell~$S = 1/(y-1)^2.$
The shell reduction yields
\[ S = D_y(S_1)+ S_1 K +  \frac{y}{(y-1)k_2},\]
where~$S_1 = 1/(1-y)$ and~$k_2 = y^2+1$. Then~$H$ can be decomposed
into~$H= D_y(S_1 T) {+} y T/ ((y-1)k_2)$, where~$T {=} \sqrt{y^2{+}1}$.
By Theorem~4 in~\cite{GeddesLeLi2004},~$H$ is not integrable.
\end{example}

\smallskip \noindent
On the other hand, it is possible that~$a$ in~\eqref{EQ:shell} is nonzero but~$H$ is integrable.
\begin{example} \label{EXAM:integrable1}
Let~$H=y \exp(y)$ whose logarithmic derivative has kernel~$1$ and shell~$y$, that is,~$H=y \exp\left(\int 1 dy \right)$.
But~$H$ is integrable as it is equal to~$D_y\left( y \exp(y) - \exp(y) \right)$.
\end{example}

\smallskip \noindent
The shell reduction cannot be directly used to decide
hyperexponential integrability.
To amend this, the solution proposed in \cite[Algorithm {\bf ReduceHyperexp}]{GeddesLeLi2004} was to find the polynomial solutions
of an auxiliary first-order linear differential equation.
In the following section, we show how this can be avoided and improved.

\section{Hermite reduction for \\ hyperexponential functions} \label{SECT:hermiteH}
After the shell reduction described in~\eqref{EQ:shell}, the denominators of
shells have been reduced to squarefree polynomials.
In the rational case, i.e., when the kernel~$K$ is zero, the polynomial~$a$ in~\eqref{EQ:certred}
can be chosen s.t.~$\deg(a)<\deg(b)$, because
all polynomials are rational integrable. But a hyperexponential function with a polynomial shell
is not necessarily integrable. For example,~$H= \exp \left(y^2\right).$

We present a differential variant of~\cite[Theorem 7]{Abramov2001} to bound the degree of~$a$ in~\eqref{EQ:certred}.
The variant leads not only to a canonical additive decomposition of hyperexponential functions, but also a direct way to
decide their integrability.

\subsection{Polynomial reduction} \label{SUBSECT:poly}
With Convention~\ref{CON:kernel}, we define
\[ \cM_K = \{ k_2 D_y (p) + k_1 p \mid  p \in \bE[y] \}. \]
It is an $\bE$-linear subspace in~$\bE[y]$. We call~$\cM_K$ the {\em subspace for polynomial reduction w.r.t.~$K$.}
Moreover, define an $\bE$-linear map~$\phi_K$ from~$\bE[y]$ to~$\cM_K$ that, for every~$p \in \bE[y]$,
maps~$p$ to~$k_2 D_y(p) + k_1 p$. We call~$\phi_K$ the {\em map for polynomial reduction w.r.t.~$K$}.

Concerning the subspace~$\cM_K$ and the map~$\phi_K$, we have
\begin{lemma} \label{LM:poly}
(i)
If~$k_2 D_y(g) + k_1 g \in \bE[y]$ for some~$g  \in  \bE(y),$ then~$g \in \bE[y]$.
(ii)
The map~$\phi_K$ is bijective.
\end{lemma}
\begin{proof}
Assume that~$g$ has a pole. Without loss of generality, we assume that the pole is $y=0$ and has order~$m$,
because the following argument is also applicable over the algebraic closure of~$\bE$.
Expanding~$g$ around the origin yields
\[ g = \frac{r}{y^m} + \mbox{terms of higher orders in~$y$}, \]
where~$r\in \bE \setminus \{0\}$. It follows from~$k_2 D_y(g) + k_1 g \in \bE[y]$ that $y = 0$ is a pole of
\[ \left( - \frac{m r}{y^{m+1}} + \mbox{higher terms} \right) + K \left(\frac{r}{y^m} + \mbox{higher terms }\right) \]
with order no more than that of~$K$. This implies that~$y{=}0$ is a simple pole of~$K$ with residue~$m$, which is incompatible
with~$K$ being differential-reduced. The first assertion holds.

The map~$\phi_K$ is surjective by its definition.
If~$\phi_K(p) {=} 0$ for some nonzero polynomial~$p {\in} \bE[y]$,
then~$K$ equals~$-D_y(p)/p$, which is nonzero since~$K \neq 0$. So~$K$ is not differential-reduced,
a contradiction. The second assertion holds.
\end{proof}

An $\bE$-basis of~$\cM_K$ is called an {\em echelon basis} if distinct elements in the basis have distinct degrees.
Echelon bases always exist and their degrees form a unique subset of~$\bN$.  Let~$\cB$ be an echelon basis of~$\cM_K$. Define
\[ \cN_K = \spanning_\bE \left\{ x^\ell \mid \mbox{$\ell \in \bN$ and $\ell \neq \deg(f)$ for all~$f \in \cB$} \right\}. \]
Then~$\bE[y] = \cM_K \oplus \cN_K$. We call~$\cN_k$ the {\em standard complement of~$\cM_K$}.
Using an echelon basis of~$\cM_K$, one can reduce a polynomial~$p$ to
a unique polynomial~$\tilde p \in \cN_K$ s.t.~$p-\tilde p \in \cM_K$.

In order to find an echelon basis of~$\cM_K$, we set~$d_1 { = } \deg k_1$, $d_2  {=} \deg k_2$,
$\tau_K { = } - \lc(k_1)/\lc(k_2)$, and $\cB { = } \left\{ \phi_K\left(y^n \right) | n \in \bN \right\}$.
By Lemma~\ref{LM:poly} (ii),~$\cB$ is an $\bE$-basis of~$\cM_K$. Let~$p$ be a nonzero polynomial in~$\bE[y]$. We make the following case distinction.

\smallskip \noindent
{\em Case 1.} $d_1 \ge d_2$. Then
$$ \phi_K(p) = \lc(k_1) \lc(p) y^{d_1 + \deg p} + \mbox{lower terms.}$$
So~$\cB$ is  an echelon basis,
in which~$\deg \phi_K\left(y^n \right) {=} d_1 {+} n$ for all~$n \in \bN$.
Accordingly,~$\cN_K$ is spanned by~$1, y, \ldots, y^{d_1 - 1}$.

\smallskip \noindent
{\em Case 2.} $d_1 = d_2 -1$ and~$\tau_K$ is not a positive integer. Then
\begin{equation} \label{EQ:case2}
\phi_K(p) {=} \left( \deg (p) \lc(k_2) {+} \lc(k_1) \right) \lc(p) y^{d_1+\deg p} {+} \mbox{lower terms}.
\end{equation}
Since~$\tau_K$ is not a positive integer, $\deg \phi_K\left(y^n \right) =  d_1 + n$.
Thus,~$\cM_K$ and~$\cN_K$ have the same bases as in Case~1.

\smallskip \noindent
{\em Case 3.} $d_1 < d_2 - 1$. If~$\deg(p)>0$, then
$$ \phi_K(p) = \deg (p) \lc(k_2) \lc(p) y^{d_2 + \deg (p) - 1} + \mbox{lower terms}.$$
Otherwise,~$\deg p = 0$ and~$\phi_K(p) = k_1 p$. Therefore,~$\cB$ is again an echelon basis,
in which
$$\deg \phi_K(1)=d_1 \,\, \text{and} \,\, \deg \phi_K(y^n)=d_2+n-1~\mbox{for all~$n \ge 1$}.$$
Accordingly,~$\cN_K$ has a basis~$1, \ldots, y^{d_1-1}, y^{d_1+1}, \ldots, y^{d_2-1}$.

\smallskip \noindent
{\em Case 4.} $d_1=d_2-1$ and~$\tau_K$ is a positive integer.
It follows from~\eqref{EQ:case2} that~$\deg \phi\left(y^n\right)=d_1+n$ if~$n \neq \tau_K$.
Furthermore, for every polynomial~$p$ of degree~$\tau_K$, $\deg(\phi_K(p))$ is of degree less than~$d_1+\tau_K$.
So any echelon basis of~$\cM_K$ does not contain a polynomial of degree~$d_1 + \tau_K$.
Set
$$\cB^\prime =\left\{ \phi\left(y^n \right) | n \in \bN, n \neq \tau_K \right\}.$$
Reducing~$\phi\left(y^{\tau_K}\right)$ by the polynomials in~$\cB^\prime$, we obtain a polynomial~$r$ of degree less
than~$d_1$. Note that~$r$ is nonzero, because~$\cB$ is an $\bE$-linearly independent set.
Hence,~$\cB^\prime \cup \{r\}$ is an echelon basis of~$\cM_K$. Consequently,~$\cN_K$ has an
$\bE$-basis~$\left\{ 1, y, \ldots, y^{\deg(r)-1}, y^{\deg(r)+1}, \ldots, y^{d_1-1}, \,  y^{d_1+\tau_K} \right\}$.
\begin{example} \label{EX:wild}
Let~$K=-6y^3/(y^4+1)$, which is differential-reduced. Then~$\tau_K=6$.
According to Case~4,~$\cM_K$ has an echelon basis
$\left\{y\} {\cup} \{ (n-6)y^{n+3} {+} ny^{n-1} | n {\in} \bN, n {\neq} 6\right\}.$
Moreover,~$\cN_K$ has a basis~$\{1, y^2, y^9\}$.
\end{example}

One can reduce the degree and number of terms of a polynomial using a subspace of polynomial reduction.
\begin{lemma} \label{LM:degree}
With Convention~\ref{CON:kernel}, we further let~$d_1 = \deg k_1$, $d_2 = \deg k_2$, and~$\tau_K=-\lc(k_1)/\lc(k_2)$.
Let~$\cM_K$ be the subspace for polynomial reduction, and~$\cN_K$ its standard complement w.r.t.~$K$.
Finally, let~$p$ be a polynomial in~$\bE[y]$.
\begin{itemize}
\item[(i)] If~$d_1 \ge d_2$ or $d_1 = d_2 - 1$ and~$\tau_K \notin \bZ^+$,
then there exists~$q \in \cN_K$ s.t.~$p \equiv q\mod\cM_K$ and~$\deg q < d_1$.
\item[(ii)] If~$d_1 < d_2 - 1$, then there exists~$q \in \cN_K$ s.t.~$p \equiv q$ $\mod$ $\cM_K$, $\deg q < d_2$
and the coefficient of~$y^{d_1}$ in~$q$ is equal to zero.
\item[(iii)] If~$d_1 {=} d_2-1$ and~$\tau_K {\in} \bZ^+$, then
there exists~$r \in \bE[y]$ of degree less than~$d_1$ s.t.~
\[ p \equiv  s y^{d_1 + \tau_K} + r \mod \cM_K \quad \mbox{for some~$s \in \bE$.} \]
Moreover,~$s y^{d_1 + \tau_K} + r$ belongs to~$\cN_K$, and~$r$ has at most $d_1-1$ terms.
\end{itemize}
\end{lemma}
\begin{proof}
The lemma is immediate from the $\bE$-bases of~$\cN_K$ constructed in the above case distinction.
\end{proof}
The next corollary enables us to derive an order bound on telescopers for hyperexponential functions.
\begin{corollary} \label{COR:term}
With the notation introduced in Lemma~\ref{LM:degree}, there exists $\cP \subset \left\{ y^n | n \in \bN \right\}$
with $|\cP| \le \max(d_1, d_2-1)$
s.t.\ every polynomial in~$\bE[y]$ can be reduced modulo~$\cM_K$ to an $\bE$-linear combination of the elements in~$\cP$.
\end{corollary}
\begin{proof}
By the above case distinction, the dimension of~$\cN_K$ over~$\bE$ is at most~$\max(d_1, d_2-1)$.
The corollary follows.
%If~$d_1 \ge d_2$ or $d_1 = d_2 - 1$ and~$\tau_K \notin \bZ^+$, then we choose~$\cP=\left\{1, y, \ldots, y^{d_1-1}\right\}$
%by Lemma~\ref{LM:degree} (i).
%If~$d_1 < d_2 - 1$, then~$\cP = \left\{1, \ldots, y^{d_1-1}, y^{d_1+1}, \ldots, y^{d_2-1} \right\}$
%by Lemma~\ref{LM:degree} (ii). If~$d_1 = d_2-1$ and~$\tau_K$ is a positive integer,
%then~$\cN_K$ has an $\bE$-basis
%$\cP = \left\{ 1, y, \ldots, y^{\deg(r)-1}, y^{\deg(r)+1}, \ldots, y^{d_1-1}, y^{\tau_K} \right\}$
%by the construction in Case~4. In each case, we see that~$|\cP|$ is no more than~$\max(d_1, d_2-1)$.
\end{proof}

\subsection{Hyperexponential integrability}
With Convention~\ref{CON:kernel}, we further assume that the polynomials~$a$ and~$b$
are obtained by the shell reduction in~\eqref{EQ:shell}.
So the decomposition~\eqref{EQ:shell} holds for the present notation.
Moreover, let~$\cM_K$ be the subspace of polynomial reduction w.r.t.~$K$, and~$\cN_K$ its standard complement.

We are going to determine necessary and sufficient conditions on hyperexponential integrability.
Since~$\gcd(b, k_2){=}1$,
\begin{equation} \label{EQ:pfd}
\frac{a}{bk_2} = p+ \frac{q}{b} + \frac{ r}{k_2},
\end{equation}
where~$p, q, r \in \bE[y]$, $\deg(q) < \deg(b)$, and~$\deg(r) < \deg(k_2)$.
Using an echelon basis of~$\cM_K$, we compute~$u$ in~$\cM_K$ and~$v$ in~$\cN_K$ s.t.~$k_2 p + r = u + v.$
By the definition of~$\cM_K$, there exists~$w$ in~$\bE[y]$ s.t.~$u=k_2D_y(w)+k_1w$.
By~\eqref{EQ:pfd}, we get
\[ \frac{a}{bk_2} =  \frac{q}{b}  + \frac{k_2D_y(w)+k_1w + v}{k_2}
                  =   D_y(w) + K w + \frac{q}{b} + \frac{v}{k_2}. \]
It follows that
\begin{equation} \label{EQ:poly}
 \frac{a}{b k_2} T = D_y\left( w T \right) +
\left(\frac{q}{b} + \frac{v}{k_2}\right) T.
\end{equation}
The process for obtaining~\eqref{EQ:poly} is referred to  as the {\em polynomial reduction
for~$(a/(bk_2)) T$ w.r.t.~$K$}, as it makes essential use of the subspaces~$\cM_K$ and~$\cN_K$.
By~\eqref{EQ:poly} and~\eqref{EQ:shell},
\begin{equation} \label{EQ:shell1}
 H = D_y((S_1 + w) T) + \left(\frac{q}{b} + \frac{v}{k_2}\right) T,
 \end{equation}
 which motivates us to introduce the notion of residual forms.
\begin{defi} \label{DEF:residual}
With Convention~\ref{CON:kernel}, we further let~$f$ be a rational function in~$\bE(y)$.
Another rational function $r {\in} \bE(y)$ is said to be a {\em residual form} of~$f$ w.r.t.~$K$
if~there exist~$g$ in~$\bE(y)$ and~$q, b, p$ in~$\bE[y]$ s.t.
\[   f = D_y(g) + K g  + r \quad \text{and} \quad r = \frac{q}{b} + \frac{v}{k_2}, \]
where~$b$ is squarefree, $\gcd(b, k_2)=1$, $\deg q < \deg b$, and~$v$ is in the standard complement~$\cN_K$ of
the subspace of polynomial reduction w.r.t.~$K$. For brevity, we say that~$r$ is {\em a residual form}
w.r.t~$K$ if~$f$ is clear from context.
\end{defi}

Residual forms are closely related to the integrability of hyperexponential functions.
\begin{lemma} \label{LM:rfi}
With Convention~\ref{CON:kernel}, we further assume that~$r$ is a nonzero residual form w.r.t.~$K$.
Then the hyperexponential function~$r T$ is not integrable.
\end{lemma}
\begin{proof}
Let~$\cM_K$ be the subspace for polynomial reduction, and~$\cN_K$ its standard complement  w.r.t.~$K$.
By the definition of residual forms, there exist~$b, q \in \bE[y]$ with~$b$ being squarefree and~$v \in \cN_K$ s.t.
\begin{equation} \label{EQ:conds}
\deg b > \deg q, \,\, \gcd(b, k_2) = 1, \,\, \text{and} \,\, r = \frac{q}{b} + \frac{v}{k_2}.
\end{equation}
Thus,~$r$ can be rewritten as~$a/(bk_2)$ for some~$a \in \bE[y]$.
Note that~$a$ is not necessarily coprime with~$bk_2$.
It follows that
\[   r T  = \frac{a}{b} \exp \left( \int \frac{k_1- D_y(k_2)}{k_2}\ dy \right). \]
Since~$(k_1- D_y(k_2))/k_2$ is differential-reduced and~$k_2, b$ are coprime,
$(a/b, (k_1 - D_y(k_2))/k_2)$ is indecomposable according to Definition~2 in~\cite{GeddesLeLi2004}.
By Theorem~4 in~\cite{GeddesLeLi2004},~$a/b$
is in~$\bE[y]$. So the denominator of~$r$ divides~$k_2$, which,
together with~\eqref{EQ:conds}, implies that~$q=0$.
Consequently,~$(v/k_2) T$ is integrable. By~\eqref{EQ:integrable},
$v = k_2 D_y \left( s \right) + k_1 s$ for some~$s \in \bE(y)$.
Since~$v \in \bE[y]$, $v \in \cM_K$ by Lemma~\ref{LM:poly}~(i).
Thus,~$v=0$ because~$ v \in \cN_K$. We have that~$r=0$, a contradiction to the assumption that~$r \neq 0$.
\end{proof}
The existence and uniqueness of residual forms are described below.
\begin{lemma} \label{LM:unique}
With~Convention~\ref{CON:kernel}, we have that the shell~$S$ has a residual form w.r.t.\ the kernel~$K$.
If a rational function has two residual forms w.r.t.~$K$, then they  are equal.
\end{lemma}
\begin{proof}
By~\eqref{EQ:shell1}, $S = D_y(S_1+w) + (S_1+w) K + q/b + v/k_2.$
So~$q/b+v/k_2$ is a required form.

Let~$r=q/b+v/k_2$ and ~$r^\prime = q^\prime/b^\prime + v^\prime/k_2$ be two residual forms of a rational function
w.r.t.~$K$, where~$b, b^\prime, q, q^\prime, v, v^\prime$ are  in~$\bE[y]$, $b$ and $b^\prime$ are squarefree,
$\gcd(b, k_2){=}\gcd(b^\prime,k_2){=}1$, $\deg q < \deg b$,
$\deg q^\prime < \deg b^\prime$ and~$v, v^\prime \in \cN_K$.
By the definition of residual forms,
$D_y(f) + f K + r = D_y(f^\prime) + f^\prime K + r^\prime$
for some~$f, f^\prime \in \bE(y)$. It follows that
\[  D_y\left(f-f^\prime \right) + \left(f -f^\prime \right) K + r - r^\prime = 0. \]
Hence,~$(r^\prime-r)T$
is integrable by~\eqref{EQ:integrable}. Since~$r-r^\prime$ is also a residual form w.r.t.~$K$,
$r=r^\prime$ by Lemma~\ref{LM:rfi}.
\end{proof}
Below is the main result of the present section.
\begin{theorem} \label{TH:hermite}
Let~$H$ be a hyperexponential function whose logarithmic derivative has kernel~$K$ and shell~$S$.
Then there is an algorithm for computing a rational function~$h$ in~$\bE(y)$
and a unique residual form~$r$ w.r.t.~$K$ s.t.
\begin{equation} \label{EQ:hermiteH}
 H = D_y \left( h \exp \left(\int K \ dy \right) \right) + r  \exp \left(\int K \ dy \right).
\end{equation}
Moreover, $H$ is integrable if and only if~$r=0$.
\end{theorem}
\begin{proof}
Let~$T=\exp \left( \int K\, dy \right)$.
Applying the shell reduction to~$H$ w.r.t.~$K$, we can find a rational function~$S_1$,
and two polynomials~$a, b$ s.t.~\eqref{EQ:certred} holds. Then we apply the polynomial
 reduction to~$a/(bk_2) T$ to get the residual form~$r = q/b+v/k_2$ s.t.~\eqref{EQ:hermiteH} holds.

Suppose that there exists another decomposition
\begin{equation} \label{EQ:hermiteH1}
H = D_y \left( h^\prime T \right) + r^\prime  T
\end{equation}
for some~$h^\prime \in \bE(y)$ and~$r^\prime$ is a residual form w.r.t.~$K$.
Then both~$r$ and~$r^\prime$ are residual forms of~$S$ by~\eqref{EQ:hermiteH}, \eqref{EQ:hermiteH1} and the fact~$H = S T$.
So~$r = r^\prime$ by Lemma~\ref{LM:unique}.

 If~$r=0$, then~$H$ is obviously integrable. Conversely, assume that~$H$ is integrable. Then~$r  T$
 is also integrable by~\eqref{EQ:hermiteH}. So~$r=0$ by Lemma~\ref{LM:rfi}.
\end{proof}

The  reduction algorithm described in the proof of Theorem~\ref{TH:hermite} has three interesting features. First, it enables us to decide
hyperexponential integrability immediately. Second, it decomposes a hyperexponential function into a sum of
an integrable one and a non-integrable one in a canonical way. Third, it does not need to compute a polynomial solution
of any first-order linear differential equation. The method will be referred to as
{\em Hermite reduction for hyperexponential functions} in the sequel, because it extends all important conclusions
obtained by Hermite reduction for rational functions to hyperexponential ones.

\begin{example} \label{EX;integrable2}
Let~$H$ be the same hyperexponential function as in Example~\ref{Exam:example1}.
Then~$K=y/(y^2+1)$ and~$S=1/(y-1)^2$. Set~$T=\sqrt{y^2+1}$.
By the shell reduction in Example~\ref{EXAM:shell},
\[
H = D_y\left(\frac{-1}{y-1} T\right)+ \frac{y}{bk_2} T,
\]
where~$b=y-1$ and~$k_2 = y^2+1$.
The polynomial reduction yields
$(y/(bk_2)) T = D_y\left( -T /2 \right) + \left(1/(2b) + 1/(2k_2) \right) T.$
Combining the above equations, we decompose~$H$ as
\[H = D_y\left(\frac{-(y+1)}{2(y-1)}T \right)+ \left(\frac{1}{2b} + \frac{1}{2k_2}\right)T.\]
%By Theorem~\ref{TH:hermite},~$H$ is not integrable.
\end{example}
\begin{example} \label{EX:integrable2}
Consider~$H=y\exp(y)$ as given in Example~\ref{EXAM:integrable1}. Since its logarithmic derivative has kernel~$K=1$,
the subspace~$\cM_K$ of polynomial reduction is equal to~$\bE[y]$. Thus, $y \in \cM_K$ and~$H$ is integrable.
More generally,~$\cM_K=\bE[y]$ corresponds to the wellknown fact that~$p(y)\exp(y)$ is integrable for all~$p \in \bE[y] \setminus \{0\}$.

%Let us consider the function~$H=\exp(y^2)$.
%Its logarithmic derivative has kernel~$K=2y$ and shell~$S=1$.
%By Case~1 in Section~\ref{SUBSEC:poly},
%the standard complement $\cN_K$ for polynomial reduction w.r.t.~$K$ is equal to~$\bE$, which contains~$S$.
%Thus,~$H$ is not integrable.
%are~$-2y^3/(y^4+1)$ and~$1$, respectively.
%The subspace~$\cM_K$ for polynomial reduction is already given in Example~\ref{EX:wild}.
%Applying Hermite reduction to~$H$ amounts to applying
%polynomial reduction to~$k_2$, because~$S=1$.
\end{example}

\section{Kernel reduction} \label{SECT:kernel}

Let~$K=k_1/k_2$ be a nonzero differential-reduced rational function in~$\bE(y)$ with~$\gcd(k_1, k_2)=1$.
We may want to reduce a hyperexponential function in the form
$$\frac{p}{k_2^m} \exp \left( \int K\ dy \right) \quad \mbox{for some~$p \in \bE[y]$ and~$m \in \bN$}.$$
One way would be to rewrite the above function as
\[  p \exp \left( \int \frac{k_1 - m D_y(k_2)}{k_2} \ dy \right), \]
and proceed by polynomial reduction w.r.t.\ the new kernel~$(k_1 - m D_y(k_2))/k_2$, which is also differential-reduced.
However, it will prove to be more convenient in Section~\ref{SECT:ct} to reduce the given function w.r.t.~the initial kernel~$K$.
To this end, we introduce another type of reduction, based on the ideas in~\cite{Davenport1986, XinZhang2009}.
\begin{lemma}\label{LM:rdshell}
With Convention~\ref{CON:kernel},
we let~$p\in \bE[y]$ and~$m \geq 1$.
Then there exist~$p_1, p_2\in \bE[y]$ s.t.
\begin{equation}\label{EQ:pbm}
\frac{p}{k_2^m} = D_y\left(\frac{p_1}{k_2^{m-1}}\right) + \frac{p_1 }{k_2^{m-1}} K + \frac{p_2}{k_2}.
\end{equation}
\end{lemma}
\begin{proof}
We proceed by induction on~$m$. If~$m =1$, then taking~$p_1=0$ and~$p_2=p$ yields the claimed form.
Assume that~$m > 1$.  We first show that there exist~$\tilde{p}_1, \tilde{p}_2 \in \bE[y]$ s.t.
\[
\frac{p}{k_2^m} = D_y\left(\frac{\tilde{p}_1}{k_2^{m-1}}\right) + \frac{\tilde{p}_1 }{k_2^{m-1}} K + \frac{\tilde{p}_2}{k_2^{m-1}},
\]
which is equivalent to
\[ p  = \tilde{p}_1(k_1-(m-1)D_y(k_2))+ (D_y(\tilde{p}_1)+ \tilde{p}_2)k_2.\]
Since~$k_1/k_2$ is differential-reduced, there exist~$u, v\in \bE[y]$ s.t.~$p = u (k_1-(m-1)D_y(k_2)) + v k_2$ by
the extended Euclidean algorithm. So we can take
~$\tilde{p}_1 {=} u$ and~$\tilde{p}_2 {=} v {-} D_y(u)$. By the induction hypothesis,
there exist~$\bar{p}_1, \bar{p}_2\in \bE[y]$ s.t.
\[ \frac{\tilde{p}_2}{k_2^{m-1}}  = D_y\left(\frac{\bar{p}_1}{k_2^{m-2}}\right) +
\frac{\bar{p}_1 }{k_2^{m-2}} K + \frac{\bar{p}_2}{k_2}.\]
Setting~$p_1 = \bar{p}_1 k_2 + \tilde{p}_1$ and~$p_2 = \bar{p}_2$ completes the proof.
\end{proof}
With Convention~\ref{CON:kernel}, we have
\[ \frac{p}{k_2^m} T  = D_y\left(\frac{p_1}{k_2^{m-1}} T \right)  + \frac{p_2}{k_2} T \]
by Lemma~\ref{LM:rdshell}. This reduction will be referred to as the {\em kernel reduction for~$(p/k_2^m) T$ w.r.t.~$K$}.

\section{Telescoping via reductions}\label{SECT:ct}
Hermite reduction has been used to construct telescopers for bivariate rational
functions in~\cite{BCCL2010}. The goal of this section is to develop a
reduction-based telescoping method for bivariate hyperexponential functions.

\subsection{Creative telescoping for bivariate rational functions}
We briefly recall the reduction-based method for rational-function
telescoping in~\cite{BCCL2010}.

Let~$\bF$ be a field of characteristic zero and~$\bF(x, y)$ be the
field of rational functions in~$x$ and~$y$ over~$\bF$. Let~$D_x$
and~$D_y$ denote the usual derivations~$\partial/\partial x$
and~$\partial/\partial y$, respectively. Let~$\bF(x)\langle D_x\rangle$ be the ring of linear differential operators
over~$\bF(x)$. The ring~$\bF(x)\langle D_x\rangle$ is a left Euclidean domain and its left ideals are principal.
For~$r\in \bF(x, y)$, the telescoping problem is to construct a
nonzero linear
differential operator~$L(x, D_x)\in \bF(x)\langle D_x\rangle$ s.t.\
$L(x, D_x)(r) = D_y(s),$
where~$s\in \bF(x, y)$. The operator~$L$ is called a \emph{telescoper} for~$r$, and~$s$
is the corresponding \emph{certificate}. The set~$\mathcal{T}$ of all telescopers for a given rational function
is a left ideal of~$\bF(x)\langle D_x\rangle$. Any generator of~$\mathcal{T}$ is called a \emph{minimal telescoper}
for the given rational function.

For any~$i\in \bN$,  rational Hermite reduction (w.r.t.~$y$) decomposes~$D_x^i(r)$ into
$D_x^i(r) { = } D_y(s_i) {+} a_i/b,$
where~$s_i {\in} \bF(x, y)$ and~$a_i, b {\in} \bF(x)[y]$ with~$\deg_y(a_i) {<} \deg_y(b)$, and~$b$ is
squarefree over~$\bF(x)$.
Since~$\deg_y(a_i)$ is bounded by~$\deg_y(b)$,  the sequence~$\{a_i\}_{i\in \bN}$ is linearly dependent over~$\bF(x)$.
Assume that there exist~$e_0, \ldots, e_{\rho} \in  \bF(x)$, not all zero, s.t.~$\sum_{i=0}^{\rho} e_i a_i=0$.
Then~$L:=\sum_{i=0}^{\rho} e_i D_x^i$ is a telescoper for~$r$ and~$\sum_{i=0}^{\rho} e_i g_i$
is the corresponding certificate. In fact, $L$ is a minimal telescoper for~$r$ if
$\rho$ is the minimal integer s.t.~$e_0, \ldots, e_{\rho}$ are linearly dependent over~$\bF(x)$.
This reasoning yields the upper bound~$\deg_y (b)$ on the order of minimal telescopers.

\subsection{Creative telescoping for bivariate hyperexponential functions}
We now apply the Hermite reduction for univariate hyperexponential functions in Section~\ref{SECT:hermiteH} to compute
telescopers for bivariate hyperexponential functions.

A nonzero element~$H$ in some differential field extension of~$\bF(x,y)$ is said to be {\em hyperexponential}
over~$\bF(x,y)$ if its logarithmic derivatives~$D_x(H)/H$ and~$D_y(H)/H$ are in~$\bF(x,y)$.

Put~$f{=}D_x(H)/H$ and~$g{=}D_y(H)/H$. Then~$D_y(f){=}D_x(g)$ because~$D_x$ and~$D_y$ commute. Therefore,
it is legitimate to denote~$H$ by~$\exp(\int  f\, dx+g \, dy)$.
For two hyperexponential functions~$H_i {=} \exp(\int  f_i \, dx+g_i \, dy)$, $i=1,2$,
we have
  \[H_1 H_2 = \exp\left(\int  (f_1+f_2)\, dx+(g_1+g_2) \, dy\right).\]
In particular, the product of a rational function~$r\in \bF(x, y)$ and a hyperexponential function~$ H = \exp(\int  f\, dx+g \, dy)$ is
\[r H = \exp\left(\int ( f + D_x(r)/r )\, dx + (g + D_y(r)/r )\, dy\right ).\]

The following fact is immediate from~\cite[Lemma 8]{GeddesLeLi2004}.
\begin{fact}\label{FACT:certdenom}
Let~$f$ and~$g$ be rational functions in~$\bF(x, y)$ satisfying~$D_y(f)=D_x(g)$.
%Then $\den(f)$ and $\den(g)$ are equal up to a multiplicative constant in $\F(x)\setminus\{0\}$.
Then the denominator of~$f$ divides that of~$g$ in~$\bF(x)[y]$.
\end{fact}

For a hyperexponential function~$H$ over~$\bF(x, y)$,
the telescoping problem is to construct a linear ordinary
differential operator $L(x, D_x)$ in~$\bF(x)\langle D_x\rangle$ s.t.
\[L(x, D_x)(H) = D_y(G)\]
for some hyperexponential function~$G$ over~$\bF(x, y)$.
As in the rational case, our idea
is to apply the Hermite reduction for univariate hyperexponential functions w.r.t.~$y$ to the derivatives~$D_x^i(H)$ iteratively,
and then find a linear dependency  among the residual forms over~$\bF(x)$.

\begin{lemma}\label{LM:iteratedhr}
Let~$ H = \exp(\int  f\, dx+g \, dy)$ be a hyperexponential function over~$\bF(x, y)$.
Let~$K$ be the kernel and~$S$ the shell of~$g$ w.r.t.~$y$.
Then, for every~$i\in \bN$, the $i$-th derivative~$D_x^i(H)$ can be decomposed into
\begin{equation}\label{EQ:ihr}
D_x^i(H) = D_y(u_i T) + r_i T,
\end{equation}
where~$u_i\in \bF(x, y)$, $T=\exp(\int  (f-D_x(S)/S)\, dx+K \, dy)$
and~$r_i \in \bF(x,y)$ is a residual form w.r.t.~$K$. Moreover, let~$k_2$ be the denominator of~$K$,
$b$~the squarefree part of the denominator of~$S$,
and~$\cN_K$ the standard complement
of the subspace for polynomial reduction w.r.t.~$K$.
Then
\begin{equation} \label{EQ:rf}
r_i=\frac{q_i}{b} + \frac{v_i}{k_2}
\end{equation}
for some~$q_i \in \bF(x)[y]$ with~$\deg_y q_i < \deg_y b$ and~$v_i \in \cN_K$.
\end{lemma}
\begin{proof}
We proceed by induction on~$i$. If~$i=0$, then the assertion holds by~Theorem~\ref{TH:hermite}.

Assume that~$D_x^i(H)$ can be decomposed into~\eqref{EQ:ihr} and assume that~\eqref{EQ:rf} holds.
Moreover, let~$\tilde{f}=f-D_x(S)/S $.
Consider the $(i+1)$-th derivative $D_x^{i+1}(h)$.
There exists a polynomial~$a$ in~$\bF(x)[y]$ s.t.~$\tilde f = a/k_2$ by $D_y\left(\tilde f \right) = D_x(K)$  and Fact~\ref{FACT:certdenom}.
A direct calculation leads to
\begin{align*}
  D_x^{i+1}(H)= & D_y(D_x(u_i T)) + \left( \frac{aq_i}{bk_2}+\frac{D_x(q_i)}{b}+\frac{D_x(v_i)}{k_2}\right)T \\
  & + \left(\frac{-q_iD_x(b)}{b^2} +  \frac{(a-D_x(k_2))v_i}{k_2^2}\right)T.
\end{align*}
Applying the shell reduction to~$\left(-q_iD_x(b)/b^2\right)T$ and the kernel reduction to~$\left((a-D_x(k_2))v_i/k_2^2\right)T$
w.r.t.~$y$, we get
\begin{align*}
  \frac{-q_iD_x(b)}{b^2} & = D_y\left(\frac{w_1}{b}\right) + \frac{w_1}{b}K + \frac{w_2}{bk_2}, \\
  \frac{(a-D_x(k_2))v_i}{k_2^2} & = D_y\left(\frac{p_1}{k_2}\right) + \frac{p_1}{k_2}K + \frac{p_2}{k_2},
\end{align*}
where~$w_1, w_2, p_1$ and~$p_2$ are in~$\bF(x)[y]$.
We then apply polynomial reduction  to~$\tilde S T$  w.r.t.~$K$, where
\[\tilde{S}= \frac{w_2}{bk_2}+\frac{p_2}{k_2}+ \frac{aq_i}{bk_2}+\frac{D_x(q_i)}{b}+\frac{D_x(v_i)}{k_2},\]
 which leads to
\[\tilde{S} = D_y(w) + wK + \left(\frac{q_{i+1}}{b} +  \frac{v_{i+1}}{k_2}\right),\]
where~$w\in \bF(x, y)$ and~$q_{i+1}/b + v_{i+1}/k_2$ is the residual form of~$\tilde{S}$
w.r.t.~$K$. It follows from a direct calculation that
\[D_x^{i+1}(H) = D_y(u_{i+1}T) + \left(\frac{q_{i+1}}{b} +  \frac{v_{i+1}}{k_2}\right) T,\]
where~$u_{i+1} = D_x(u_i) + u_i\tilde{f} + w_1/b + p_1/k_2 + w$.
\end{proof}

The main results in the present section are given below.
\begin{theorem}\label{TH:hct}
With the notation introduced in Lemma~\ref{LM:iteratedhr},
we let~$L = \sum_{i=0}^\rho e_i D_y^i$ with~$e_0, \ldots, e_\rho \in \bF(x)$, not all zero.
\begin{itemize}
\item[(i)] $L$ is a telescoper for~$H$ if and only if~$\sum_{i=1}^\rho e_i r_i = 0.$
\item[(ii)] The order of a minimal telescoper for~$H$ is no more than~$\deg_y(b) + \max(\deg_y(k_1), \deg_y(k_2)-1)$.
\end{itemize}
\end{theorem}
\begin{proof}
We set~$\bE=\bF(x)$ and view that hyperexponential functions involved in the proof are over~$\bE(y)$.
Moreover, let~$u=\sum_{i=0}^\rho e_i u_i$ and~$r = \sum_{i=0}^\rho e_i r_i.$
By~\eqref{EQ:ihr}, we have
\begin{equation} \label{EQ:tel}
L(H) = D_y(uT) +  rT.
\end{equation}
If~$r=0$, then~$L$ is a telescoper by~\eqref{EQ:tel}.
Conversely,
assume that~$L$ is a telescoper of~$h$. Then~$r T$
is integrable w.r.t.~$y$ by~\eqref{EQ:tel}. Since~$r$ is a residual form, it is equal to zero
by Lemma~\ref{LM:rfi}. The first assertion is proved.

Set~$\lambda=\max (\deg_y(k_1), \deg_y(k_2)-1)$.
Let the residual form~$r_i = q_i/b + v_i/k_2$ be as defined in~\eqref{EQ:ihr} and~\eqref{EQ:rf}.
By Corollary~\ref{COR:term}, the~$v_i$'s have a common set~$\cP$ of supporting monomials with~$|\cP|\le \lambda$.
Moreover,~$\deg_y(q_i)<\deg_y(b)$ and~$\gcd(b,k_2)=1$.
Therefore, the residual forms~$r_0,$ \ldots, $r_{\rho}$
are linearly dependent over~$\bF(x)$ if~$\rho\geq \deg_y(b) + \lambda$.
The second assertion holds
\end{proof}
\begin{remark} \label{RE:minimal}
By Theorem~\ref{TH:hct}, the first linear dependency among the residual forms~$r_0, r_1, r_2, \ldots$
gives rise to a minimal telescoper of~$H$.
\end{remark}

Below is an outline of the reduction based telescoping algorithm for hyperexponential
functions, in which the notation is that introduced in Lemma~\ref{LM:iteratedhr} is used.

\medskip \noindent
{{\bf Algorithm.}~\textsf{HermiteTelescoping}}:
Given a bivariate hyperexponential function $H=\exp(\int f\, dx+g\, dy)$ over~$\bF(x,y)$,
compute a minimal telescoper $L$ and its certificate w.r.t.~$y$.
\begin{enumerate}
\item Find the kernel~$K$ and shell~$S$ of~$D_y(H)/H$ w.r.t.~$y$. Set~$b$ to be the squarefree part of the denominator of~$S$.
\item Decompose~$H$ into~$H = D_y(u_0T) + r_0 T$
using the Hermite reduction for hyperexponential functions given in Theorem~\ref{TH:hermite}.
If $r_0=0$, return $(1, u_0T)$.
\item Set~$\rho:= \deg_y(b) + \max (\deg_y(k_1), \deg_y(k_2)-1)  $.
\item For $i$ from 0 to $\rho$ do
\begin{enumerate}
\item[4.1.] Compute~$(u_i, r_i)$ incrementally s.t.
\[D_x^i(H) = D_y(u_iT) + r_i T \]
by the shell, kernel and polynomial reductions described in Lemma~\ref{LM:iteratedhr}.
\item[4.2.] Find~$\eta_j \in \bF(x)$ s.t.~$\sum_{j=0}^i \eta_j r_j = 0$
 using the algorithm in~\cite{Storjohann2005}.
If there is a nontrivial solution, return~$\left(\sum_{j=0}^i \eta_j D_x^j, \,
\sum_{j=0}^i \eta_ju_jT \right)$.
\end{enumerate}
%\item Compute the content~$c$ of~$L$ and return~$(c^{-1}L, c^{-1}g)$.
\end{enumerate}}

%%%%%%%%%%%%%%%%%%%%%%%%%%%%%%%%%%%%%%%%
%%%%%%%%%%%%%%%%%%%%%%%%%%%%%%%%%%%%%%%%

\begin{example}
Let~$H=\sqrt{x-2y}\,\exp(x^2y)$. Then~$D_x(H)/H$ and~$D_y(H)/H$ are, respectively,
\[f = \frac{1+4x^2y-8xy^2}{2(x-2y)} \quad \text{and} \quad g=\frac{-1+x^3-2x^2y}{x-2y}. \]
Since~$g$ is differential-reduced w.r.t.~$y$, $g$ is the kernel and~$1$ is the shell of~$D_y(H)/H$ w.r.t.~$y$.
By Hermite reduction,
\begin{equation}\label{EQ:h0}
H = D_y\left(\frac{1}{x^2} H\right) +
\frac{1}{x^2 k_2} H.
\end{equation}
Applying~$D_x$ to the above equation yields
\[D_x(H) = D_y\left(\frac{-3x+8y+4x^3y-8x^2y^2}{2x^3(x-2y)} H\right) +
r H, \]
where~$r={(-5x+8y+4x^3y-8x^2y^2)}/{(2x^3 k_2^2)}$.
The shell, kernel and polynomial reduction  given in Lemma~\ref{LM:iteratedhr} yields
\begin{equation}\label{EQ:h1}
D_x(H) = D_y\left(\frac{2x^2y-3}{x^3} \cdot H \right) + \frac{3x^3-6}{2x^3k_2} H
\end{equation}
Combining~\eqref{EQ:h0} and~\eqref{EQ:h1}, we get~$L = (6-3x^3)+2xD_x$ is a minimal telescoper for~$H$
and~$G =(4y-3x)H$ is the corresponding certificate.
\end{example}
\begin{remark}
The algorithm \textsf{HermiteTelescoping} is directly based on the proof of Lemma~\ref{LM:iteratedhr}.
Yet, there is another idea for computing a minimal telescoper of~$H$. Namely, we first compute a nonzero
operator~$L_1 \in \bF(x)\langle D_x \rangle$ of minimal order s.t.~$L_1(H)=D_y(G_1) + (p/k_2) T$ for some hyperexponential function~$G_1$
and polynomial~$p$. Note that such operators always exist,
because~$\deg_y q_i$ in~\eqref{EQ:rf} is less than~$\deg_y b$. Then we apply the algorithm \textsf{HermiteTelescoping}
 to get a minimal telescoper~$L_2$ for~$(p/k_2)T$. In doing so, any rational function with denominator~$b$ will not appear when we compute~$L_2$.
It turns out that~$L_2L_1$ is a minimal telescoper of~$H$. An implementation on this idea is underway.
\end{remark}
\subsection{Comparison with the Apagodu-Zeilberger bound}
Assume that
\begin{equation} \label{EQ:mult1}
H = u \exp\left(\frac{r_1}{r_2}\right) \prod_{i=1}^m p_i(x, y)^{c_i},
\end{equation}
where~$u, r_1, r_2, p_1, \ldots, p_m$ are nonzero polynomials in~$\bF[x,y]$
and~$c_1, \ldots, c_m$ are {\em distinct indeterminates}.
Theorem~cAZ  in~\cite{Apagodu2006} asserts that the order of minimal telescopers for~$H$ is bounded by
\[ \alpha := \deg_y (r_2) +  \max\left(\deg_y (r_1), \deg_y (r_2) \right) + \sum_{i=1}^m \deg_y (p_i) - 1. \]
Note that~$H$ can be viewed as a hyperexponential function over~$\bF(c_1, \ldots, c_m)(x,y)$.
We now show that~$\alpha$ given above is no less than the order bound on minimal telescopers for~$H$  obtained from Theorem~\ref{TH:hct}~(ii).
The kernel and shell of the logarithmic derivative~$D_y(H)/H$ are
$$K := D_y\left(\frac{r_1}{r_2}\right) + \sum_{i=1}^m c_i\frac{D_y(p_i)}{p_i} \quad \text{and} \quad
S := u,$$
respectively, because~$K$ has no integral residue at any  simple pole,~$S$ is a polynomial in~$\bF[x,y]$, and~$D_y(H)/H$ is
equal to~$K+D_y(S)/S$.
Let~$K=k_1/k_2$ with~$\gcd(k_1,k_2)=1$. A direct calculation leads to
\[ \deg_y (k_1) {\le}  \deg_y(r_1) + \deg_y (r_2) +  \sum_{i=1}^m \deg_y (p_i)  - 1, \]
and
$$\deg_y (k_2) \le 2 \deg_y (r_2) +   \sum_{i=1}^m \deg_y (p_i).$$
By Theorem~\ref{TH:hct}, the order of minimal telescopers for~$H$ is no more than $\max\left(\deg_y(k_1), \deg_y(k_2)-1\right)$,
which is no more than~$\alpha$ by the above two inequalities.

Indeed, the order bound in Theorem~\ref{TH:hct} (ii) may be smaller than that in Theorem cAZ.
\begin{example} \label{EX:bound}
Let~$H{=}q^c \exp(a/q)$, where~$a, q$ are irreducible polynomials in~$\bF[x,y]$ with~$\deg_y (a) < \deg_y (q),$ and~$c$ is
a transcendental constant over~$\bF$.
By Theorem cAZ,  a minimal telescoper for~$H$ has order no more than~$3 \deg_y q {-} 1$.
On the other hand,
the kernel and shell of~$D_y(H)/H$ are equal to $\left(D_y(a)q{-}aD_y(q){+}cqD_y(q)\right)/q^2$ and~$1$, respectively.
A minimal telescoper has order no more than~$2 \deg_y q - 1$ by Theorem~\ref{TH:hct}~(ii).
\end{example}

In general, Christopher's Theorem states that a hyperexponential function over~$\bF(x,y)$ can always be written as:
\begin{equation} \label{EQ:mult}
\frac{u}{v} \exp\left( \frac{r_1}{r_2} \right) \prod_{i=1}^m p_i(x, y)^{c_i},
\end{equation}
where~$u,v, r_1, r_2 \in \bF[x,y]$, $c_i$  is algebraic over~$\bF$, and~$p_i$ is in~$\bF(c_i)[x,y]$, $i=1, \ldots, m$.
A more explicit description on~\eqref{EQ:mult} can be found in~\cite{ChenThesis}. So~$H$ given in~\eqref{EQ:mult1}
is a special instance for hyperexponential functions.
In addition, it is easier to compute the kernel and shell than to compute the decompositions~\eqref{EQ:mult1} and~\eqref{EQ:mult} when
a hyperexponential function is given by its logarithmic derivatives.

\section{Implementation and timings} \label{SECT:timings}
We have produced a preliminary implementation of the algorithm \textsf{HermiteTelescoping} in the
computer algebra system~{\sf Maple 16}.
Our Maple code is available from
\[\text{\url{http://www4.ncsu.edu/~schen21/HermiteCT.html}}\]
We now compare the performance of our algorithm to the Maple
function~{\sf DEtools[Zeilberger]} of the telescoping
algorithm in~\cite{Almkvist1990}. The examples for comparison are of the form
\[\frac{p}{q^m} \cdot \sqrt{\frac{a}{b}}\cdot \exp\left(\frac{u}{v}\right),\]
where~$m\in \bN$, $p, q, a, b, u, v\in \bZ[x, y]$ are irreducible and their coefficients
are randomly chosen. For simplicity, we choose
$\lambda = \deg_y(p) = \deg_y(q)$, $\mu= \deg_y(a)=\deg_y(b)$, and $\nu= \deg_y(u)=\deg_y(v)$.
The runtime  comparison (in seconds) for different examples is shown in Table~\ref{tab:1}, in which
\begin{itemize}
  \item {\sf ZT}: the Maple function~{\sf DEtools[Zeilberger]}.
  \item {\sf HT}: our implementation of \textsf{HermiteTelescoping}.
  \item order: the order of the computed minimal telescoper.
  \item OOM: Maple runs out of memory.
\end{itemize}

\begin{table}[h]
  \vspace{-\smallskipamount}
  \begin{center}
  \def\c#1{\hbox to4em{\hss\smash{\raisebox{0.25ex}{#1}}\hss}}
  \begin{tabular}{r|r|r|c}
    $(\lambda, \mu, \nu, m)$   & {\sf{ZT}}   & {\sf{HT}}    & \c{order}    \\ \hline
     (2, 0, 2, 1) & 2.23  &  2.43  & 5   \\ % in-43.m
     (2, 0, 2, 2) & 2.21  &  2.01  & 5   \\ % in-12.m
     (3, 0, 2, 1) & 8.72  &  6.64  & 6  \\ % in-23.m
     (3, 0, 2, 2) & 9.38  &  6.56  & 6   \\ % in-49
     (6, 0, 1, 1) & 45.35 & 24.49  & 7   \\ % in-43.m
     (6, 0, 1, 2) & 43.02 & 22.91  & 7   \\ % in-12.m
     %(2, 1, 2, 1) & 80.51 & 17.01  & 7    \\ % in-43.m
     %(2, 1, 2, 2) & 66.61 & 16.92  & 7   \\ % in-12.m
     %(2, 1, 2, 1) & 64.04 & 17.37  & 7   \\ % in-23.m
     %(2, 1, 2, 2) & 70.94 & 17.73  & 7  \\ % in-49
     (2, 2, 2, 1) & 1405.9 & 221.5  & 9   \\ % in-43.m
     (2, 2, 2, 2) & 1398.1 & 200.34 & 9   \\ % in-12.m
     (3, 0, 3, 1) & 147.92 & 47.23  & 8  \\ % in-23.m
     (3, 0, 3, 2) & 151.20 & 44.56  & 8  \\
     (3, 3, 0, 1) & 207.82 & 61.10  & 8   \\ % in-12.m
     (3, 3, 0, 2) & 211.70 & 58.63  & 8   \\ % in-23.m
     (3, 2, 1, 1) & 304.61 & 62.67  & 8   \\ % in-49
     (3, 2, 1, 2) &  331.0 & 63.61  & 8   \\ % in-43.m
     (3, 1, 3, 1) & OOM & 534.87     & 10 \\ % in-12.m
     (3, 1, 3, 2) & OOM   &  522.15  & 10  \\
     \hline
  \end{tabular}
  \end{center}
  \vspace{-\smallskipamount}
  \caption{{\small Timings (in sec.) were taken on a Mac OS X computer with 4Gb RAM and 3.06 GHz Core 2 Duo processor.}}\label{tab:1}
\end{table}
\begin{remark} \label{RE:bound}
The orders of the computed minimal telescopers in our experiments are equal to the predicted order bounds in Theorem~\ref{TH:hct} .
\end{remark}

\bibliographystyle{plain}

%\bibliography{Hermite}

\end{document}